\RequirePackage{etoolbox}
\patchcmd{\bibliographystyle}{#1}{abbrvnat}{}{}
\documentclass[tablecaption=bottom,wcp]{jmlr} %

\usepackage[font=small,skip=-6pt]{caption} %
 \usepackage{amsmath,amssymb}
 \usepackage{lipsum,graphicx,multicol}

\usepackage{dsfont}
\newcommand{\burnin}{\ell}
\newcommand{\maxK}{\widetilde{K}}
\newcommand{\GibbsCond}[2]{p_{\Pi | \Pi_{-#2}}}
\newcommand{\defined}{:=}
\newcommand{\ind}[1]{\mathds{1}\{#1\}} %
\newcommand{\calP}[1]{\mathcal{P}_{#1}} %
\newcommand{\given}{\,|\,}
\newcommand{\CRP}[2]{\text{CRP}(#1,#2)}
\newcommand{\distNorm}[2]{\mathcal{N}(#1,#2)}
\newcommand{\distDP}[2]{\mathrm{DP}(#1,#2)}
\newcommand{\data}{W}

\newcommand{\iid}{\textrm{i.i.d.}\xspace}
\newcommand{\dist}{\sim}
\newcommand{\distiid}{\overset{\textrm{\tiny\iid}}{\dist}}

 \def\[#1\]{\begin{align}#1\end{align}}        %
\def\(#1\){\begin{align*}#1\end{align*}}     %

\DeclareMathOperator*{\argmin}{arg\,min}

 \usepackage{mathtools}
 \usepackage{verbatim}
 \usepackage{microtype}
 \usepackage{graphicx}
 \usepackage{booktabs} %
 \usepackage{commath}
 \usepackage{bbm}
 \usepackage{color}
 \usepackage [autostyle, english = american]{csquotes}
 \usepackage[capitalize]{cleveref}

\makeatletter
\newcommand{\printfnsymbol}[1]{%
  \textsuperscript{\normalfont\color{blue}\@fnsymbol{#1}}%
}
\makeatother

\usepackage[load-configurations=version-1]{siunitx} %

\jmlrproceedings{AABI 2020}{3rd Symposium on Advances in Approximate Bayesian Inference, 2020}

\title[Optimal Transport Couplings for Gibbs Samplers on Partitions]{Optimal Transport Couplings of Gibbs Samplers on Partitions for Unbiased Estimation}

\author{\\
    \Name{{Brian L.} Trippe}\thanks{These authors contributed equally}\Email{btrippe@mit.edu}\\
    \Name{{Tin D.} Nguyen}\printfnsymbol{1} \Email{tdn@mit.edu}\\
 \Name{Tamara Broderick} \Email{tbroderick@csail.mit.edu}\\
 \addr MIT CSAIL}

\begin{document}

\maketitle

\begin{abstract}
	Computational couplings of Markov chains provide a practical route to unbiased Monte Carlo estimation that can utilize parallel computation.
However, these approaches depend crucially on chains meeting after a small number of transitions.
For models that assign data into groups, e.g.\ mixture models, the obvious approaches to couple Gibbs samplers fail to meet quickly.
This failure owes to the so-called `label-switching' problem; semantically equivalent relabelings of the groups contribute well-separated posterior modes that impede fast mixing and cause large meeting times.
We here demonstrate how to avoid label switching by 
considering chains as exploring the space of partitions rather than labelings.
Using a metric on this space, we
employ an optimal transport coupling of the Gibbs conditionals.
This coupling outperforms alternative couplings that rely on labelings
and, on a real dataset, provides estimates more precise than usual ergodic averages in the limited time regime.
Code is available at \href{https://github.com/tinnguyen96/coupling-Gibbs-partition}{github.com/tinnguyen96/coupling-Gibbs-partition}.

\end{abstract}

\section{Introduction} \label{sec:intro}
\paragraph{Couplings for unbiased Markov chain Monte Carlo.}
Consider estimating an analytically intractable expectation of a function $h$ of a random variable $X$ distributed according to $p$, $H^* \coloneqq \int h(X) p(X) dX$.
Given a Markov chain $\{X_t\}_{t=0}^\infty$ with initial distribution $X_0\sim p_0$ 
and evolving according to a transition kernel $X_t \sim T(X_{t-1}, \cdot)$ stationary with respect to $p$,
one option is to approximate $H^*$ with the empirical average of samples $\{h(X_t)\}$.
However, while ergodic averages are asymptotically consistent, 
they are in general biased when computed from finite simulations. 
As such, one cannot effectively utilize parallelism to reduce error to any desired tolerance.

Computational couplings provide a route to unbiased estimation in finite simulation \citep{glynn2014exact};
in this work we build on the framework of \citet{jacob2020unbiased}.
One designs an additional Markov chain $\{Y_t\}$ with two properties. 
First, $Y_t \given Y_{t-1}$ also evolves using the transition $T(\cdot, \cdot)$, so that $\{Y_t\}$ is equal in distribution to $\{X_t\}$.
Secondly, there exists a random \textit{meeting time} $\tau < \infty$ such that the two chains meet exactly at some time $\tau$, $X_{\tau}= Y_{\tau-1}$, and remain faithful afterwards: for all $t\ge\tau$, $X_{t}= Y_{t-1}$. 
Then, one can compute an unbiased estimate of $H^*$ as
\[\label{eqn:unbiased_estimate}
H_{\burnin:m}(X, Y) \coloneqq
\underbrace{\frac{1}{m-\burnin+1} \sum_{t=\burnin}^m h(X_t) }_{\text{Usual MCMC average}} + 
\underbrace{\sum_{t=\burnin+1}^{\tau - 1} \text{min}\left(1, \frac{t-\burnin}{m-\burnin+1}\right)\left\{h(X_t) - h(Y_{t-1}) \right\}}_{\text{Bias correction}} 
\]
where $\burnin$ is the burn-in length, and $m$ sets a minimum number of iterations \citep[Equation 2]{jacob2020unbiased}.
One interpretation of this estimator is as the usual MCMC estimate plus a bias correction.
Since $H_{\burnin:m}$ is unbiased, we can make the squared error (for estimating $H^*$) arbitrarily small by simply averaging many estimates computed in parallel.
However, the practicality of \Cref{eqn:unbiased_estimate} relies on a coupling that provides sufficiently small meeting times.
Large meeting times are doubly problematic:
they lead to greater computational cost and higher variance due to the additional terms.

\paragraph{Gibbs samplers over discrete structures and their couplings.}
Gibbs sampling is a standard inference method for models with discrete structures and tractable conditional distributions.  
Numerous applications include Bayesian nonparametric clustering using Dirichlet process mixture models \citep{antoniak1974mixtures,neal2000markov}, graph coloring for randomized approximation algorithms \citep{jerrum1998mathematical}, community detection using stochastic block models \citep{holland1983stochastic,geng2019probabilistic} and computational redistricting \citep{DeFord2021Recombination}. 
In these cases, the discrete structure is the partition of data into components. 

While some earlier works have described couplings of Gibbs samplers, they have not sought to address computational approaches applicable in these settings.
For example, \citet{jerrum1998mathematical} uses maximal couplings on labelings to prove convergence rates for graph coloring,
and \citet{gibbs2004convergence} uses a common random number coupling for two-state Ising models.
Notably, these approaches rely on explicit labelings and, in our experiments, suffer from large meeting times.
We attribute this issue to the label-switching problem \citep{jasra2005markov};
heuristically, many different labelings imply the same partition, and two chains may nearly agree on the partition but require many iterations to change label assignments.

\paragraph{Our contribution.}
We view the Gibbs sampler as exploring a state-space of partitions rather than labelings (as, for example, in \citet{tosh2014lower}), and define an optimal transport (OT) coupling in this space.
We show that our algorithm has a fast run time and empirically validate it in the context of Dirichlet process mixture models \citep{antoniak1974mixtures,prabhakaran2016dirichlet}
and graph coloring \citep{jerrum1998mathematical},
where it provides smaller meeting times than the label-based couplings of \citet{jerrum1998mathematical,gibbs2004convergence}.
We demonstrate the benefits of unbiasedness by reporting estimates of the posterior predictive density and cluster proportions. Our implementation is publicly available at \href{https://github.com/tinnguyen96/coupling-Gibbs-partition}{github.com/tinnguyen96/coupling-Gibbs-partition}.

\section{Our Method}\label{sec:method}

\subsection{Gibbs samplers over partitions}\label{sec:gibbs_on_partitions}
For a natural number $N$, a \emph{partition} of $[N]  \defined \{1,2,\ldots,N\}$ is a collection of non-empty disjoint sets $\{A_1,A_2,\ldots,A_k\}$, whose union is $[N]$ \citep[Section 1.2]{pitman2006combinatorial}.
We use $\calP{N}$ to denote the set of all partitions of $[N]$.
Throughout, we use $\pi$ to denote elements of $\calP{N}$ and $\Pi$ for a random partition (i.e.\ a $\calP{N}$-valued random variable) with probability mass function (PMF) $p_{\Pi}$.
Finally $\pi_{-n}$ and $\Pi_{-n}$ denote these partitions with data-point $n$ removed.
For example, if $\pi = \left\{\{1, 3\}, \{ 2\}\right\}$, then $\pi_{-1}=\left\{ \{3 \},\{ 2\} \right\}$.

Drawing direct Monte Carlo samples $\Pi \sim p_{\Pi}$ is often impossible.
However, the conditional distributions $\GibbsCond{N}{n}$ are supported on at most $N$ partitions. 
Hence, when $p_{\Pi}$ is available up to a proportionality constant, computing and sampling from $\GibbsCond{N}{n}$ are tractable operations.
A Gibbs sampler exploiting this tractability proceeds as follows.
First, a partition $\pi$ is drawn from an initial distribution $p_0$ on $\calP{N}$. For each iteration, we \emph{sweep} through each data-point $n\in [N]$, temporarily remove it from $\pi$, and then randomly reassign it to one of the sets within $\pi_{-n}$ or add it as singleton (that is, as a new group) according the conditional PMF $\GibbsCond{N}{n}(\cdot | \pi_{-n})$.

\subsection{Our approach: optimal coupling of Gibbs conditionals} 
\begin{algorithm2e}[!t]
\caption{Gibbs Sweep with Optimal Transport Coupling}
\label{alg:gibbs}
 \LinesNumbered
 \KwIn{Target probability mass function (PMF) $p_\Pi$.  Current partitions $\pi$ and $\nu$.}
\For{$n\leftarrow 1$ \KwTo $N$}{
    // Compute Gibbs marginals (PMFs over partitions) \\
    $q, r \leftarrow \GibbsCond{N}{n}(\cdot | \pi_{-n}), \GibbsCond{N}{n}(\cdot | \nu_{-n})$ \\
    \text{ }\\ %
    // Compute and sample from optimal transport coupling\\
    $[\pi^1, \pi^2, \dots, \pi^K],[\nu^1, \nu^2, \dots, \nu^{K^\prime}] \leftarrow \text{support}(q), \text{support}(r)$ \\
    $\gamma^* = \argmin_{\gamma \in \Gamma(q, r)} \sum_{k=1}^K \sum_{k^\prime=1}^{K^\prime} \gamma(\pi^k, \nu^{k^\prime}) \text{d}(\pi^k, \nu^{k^\prime})$ \\
    $\pi, \nu  \sim \gamma^*$\\
}
\text{Return} $\pi, \nu$
\end{algorithm2e}

Our coupling encourages the chains to become `closer' while maintaining the correct marginal evolution.
To quantify closeness we use a metric on $\calP{N}$.
While a number of metrics exist \citep[Section 2]{meilua2007comparing}, for simplicity we chose a classical metric introduced by \citet{mirkin1970measurement,rand1971objective},
\[\label{eqn:distance_between_partitions}
    \text{d}(\pi, \nu) 
    = \sum_{A\in \pi}|A|^2 
      +  \sum_{B \in \nu} |B|^2
      -2\sum_{A\in \pi, B\in\nu} |A\cap B|^2,
\]
which is equivalent to Hamming distance on the adjacency matrices implied by partitions \citep[Theorems 2-3]{mirkin1970measurement}. 
We leave investigation of the impact of metric choice on meeting time distribution to future work. 

With the metric in \Cref{eqn:distance_between_partitions}, we can formalize an \emph{optimal transport coupling} of two Gibbs conditionals, i.e.\ the coupling that
minimizes the expected distances between the updates.
In particular, we let $q:=\GibbsCond{N}{n}(\cdot | \pi_{-n})$ and 
$r:=\GibbsCond{N}{n}(\cdot | \nu_{-n})$ with supports 
$[\pi^1, \pi^2, \dots, \pi^K] := \text{support}(q)$ and
$[\nu^1, \nu^2, \dots, \nu^{K^\prime}] := \text{support}(r)$
and define the OT coupling as
\[\label{eq:minimal-Ed-coupling}
    \gamma^* :=  \argmin_{\gamma \in \Gamma(q, r)} \sum_{k=1}^K \sum_{k^\prime=1}^{K^\prime} \gamma(\pi^k, \nu^{k^\prime}) \text{d}(\pi^k, \nu^{k^\prime}),
\]
where $\Gamma(q, r)$ is the set of all couplings of $q$ and $r$.
\Cref{alg:gibbs} summarizes this approach.

\subsection{Efficient computation of optimal couplings}
The practicality of our OT coupling depends both on successfully encouraging chains to meet in a small number of steps
and on an implementation with computational cost comparable to running single chains.
If \Cref{alg:gibbs} required orders of magnitude more time than the Gibbs sweep of single chains,
the extent of parallelism required to place the unbiased estimates from coupled chains on an even footing with standard MCMC could be prohibitive.

In many applications, including those in our experiments, for partitions of size $K$, the Gibbs conditionals may be computed in $\Theta(K)$ time,
and a full sweep through the $N$ data-points takes $\Theta(NK)$ time for a single chain.
At first consideration, an implementation of \Cref{alg:gibbs} with comparable efficiency might seem infeasible.
In particular, when $\pi$ and $\nu$ are of size $O(K),$ \Cref{eq:minimal-Ed-coupling} requires computing $O(K^2)$ pairwise distances,
each of which naively might seem to require at least $O(KN)$ operations --- let alone the OT problem.

The following result shows that we can in fact compute this coupling efficiently.
\begin{theorem}[Gibbs Sweep Time Complexity]\label{thm:Gibbs_complexity}
Let $p_{\Pi}$ be the law of a random $N$-partition.
If for any $\pi \in \calP{N}, \GibbsCond{N}{n}(\cdot | \pi_{-n})$ is computed in constant time, 
the Gibbs sweep in \Cref{alg:gibbs} has $O(N\maxK^3 \log \maxK)$ run time, where $\maxK$ is the max partition size encountered.
\end{theorem}
As a proof of \Cref{thm:Gibbs_complexity}, we detail an $O(N\maxK^3\log \maxK)$ implementation in \Cref{apd:gibbs_sweep_proof}.

\Cref{thm:Gibbs_complexity} guarantees that the run time of a coupled-sweep is no more than a $O(\maxK^2 \log \maxK)$ factor slower than a single-sweep.
The relative magnitude of $\maxK$ versus $N$ depends on the target distribution.
For the graph coloring distribution, $\maxK$ is upper bounded by the numbers of available colors. 
Under the Dirichlet process mixture model (DPMM) prior, with high probability, the size of partition of $N$ data points is within multiplicative factors of $\ln N$ \citep[Section 5.2]{arratia2003logarithmic}.
We \emph{conjecture} that under most initializations of the Gibbs sampler (such as from the DPMM prior), $\maxK = O(\ln N)$ with high probability. 

\begin{remark}\label{rem:complexity}
The worst-case run time of \Cref{thm:Gibbs_complexity} is attained with Orlin's algorithm \citep{orlin1993faster} to solve \Cref{eq:minimal-Ed-coupling} in $O(\maxK^3\log \maxK)$ time.
However, our implementation uses the simpler network simplex algorithm \citep{kelly1991minimum} as implemented by \citet{flamary2021pot}.
Although \citet[Section 3.6]{kelly1991minimum} upper bound the worst-case complexity of the network simplex as $O(\maxK^5)$, the algorithm's average-case performance may be as good as $O(\maxK^2)$ \citep[Figure 6]{bonneel2011displacement}.
\end{remark}

Although Orlin's algorithm \citep{orlin1993faster} has a better worst-case runtime, convenient public implementations are not available. 
In addition, our main contribution is the formulation of the coupling as an OT problem --- in principle, the dependence on $\maxK$ of the runtime in \Cref{thm:Gibbs_complexity} inherits from the best OT solver used.

\section{Empirical Results}\label{sec:applications}
In \Cref{sec:meeting-times}, we compare the distribution of meeting times between our partition-based coupling and two label-based couplings: under our coupling, chains meet earlier.
In \Cref{sec:unbiased-estimation}, we report unbiased estimates of two estimands of common interest: posterior predictive densities and the posterior mean proportion of data assigned to the largest clusters.
But first, we describe the applications and the target distributions under consideration in \Cref{sec:targets}. 

\subsection{Applications} \label{sec:targets}
\paragraph{Dirichlet process mixture models.} 
Clustering is a core task for understanding structure in data and density estimation.
When the number of latent clusters is a priori unknown, DPMMs \citep{antoniak1974mixtures} are a useful tool.
The cluster assignments of data points in a DPMM can be described with a Chinese restaurant process, or $\CRP{\alpha}{N}$, which is a probability distribution over $\calP{N}$ with mass $\Pr( \Pi = \pi ) = \frac{\alpha^{K} \prod_{A \in \pi} (|A| -1)! }{\alpha (\alpha + 1) \ldots (\alpha + N-1)}$ where $K$ is the number of clusters in $\pi$, and $\prod_{A \in \pi}$ iterates through the clusters.
We consider a fully conjugate DPMM \citep{maceachern1994estimating},
\begin{equation} \label{eq:dpmm-model}
\Pi \sim \CRP{\alpha}{N}, \hspace{10pt} \mu_A \distiid \distNorm{\mu_0}{\Sigma_0} \text{ for } A \in \Pi, \hspace{10pt}
\data_{j} \given  \mu_A \distiid \distNorm{\mu_A}{\Sigma_1} \text{ for } j \in A.
\end{equation}
The hyper-parameters of \Cref{eq:dpmm-model} are concentration $\alpha$, cluster prior mean $\mu_0$, observational covariance $\Sigma_1$ and cluster covariance $\Sigma_0$.
For this application, the distribution is the Bayesian posterior, $p_\Pi(\pi) := \Pr(\Pi=\pi\given \data)$.
The Gibbs conditionals of the posterior $\GibbsCond{N}{n}$ can be computed in closed form, using simple formulas for conditioning of jointly Gaussian random variables and the well-known Polya urn scheme \citep[Equation 3.7]{neal2000markov}. 

\paragraph{Graph coloring.} Uniform sampling of graph colorings is a problem of fundamental interest in theoretical computer science for its role as a subroutine within fully polynomial randomized approximation algorithms, 
where samples from the uniform distribution on graph colorings are used to estimate the number of unique colorings \citep{jerrum1998mathematical}.

Notably, this sampling problem reduces to sampling from the induced distribution on partitions,
by choosing an ordering of the sets in the partition and associating it with a random permutation of the set of colors.
Accordingly, estimates are just as easily constructed for a Markov chain defined on partitions.
See \Cref{apd:experimental_setup} for additional details.

\subsection{Reduced meeting times with OT couplings} \label{sec:meeting-times}
\begin{figure}[t!]
    \floatconts
    {fig:toy_results}
    {\caption{
    Reduced meeting times are achieved by OT couplings of Gibbs conditionals relative to maximal and common random number couplings
    in applications to (A) DPMM and (B) graph coloring.
    (A) Left and (B) left show two representative traces of the distance between coupled chains by iteration. (A) Right and (B) right show histograms of meeting times 250 replicate coupled chains.
    }}
    {
     \includegraphics[width=0.95\linewidth]{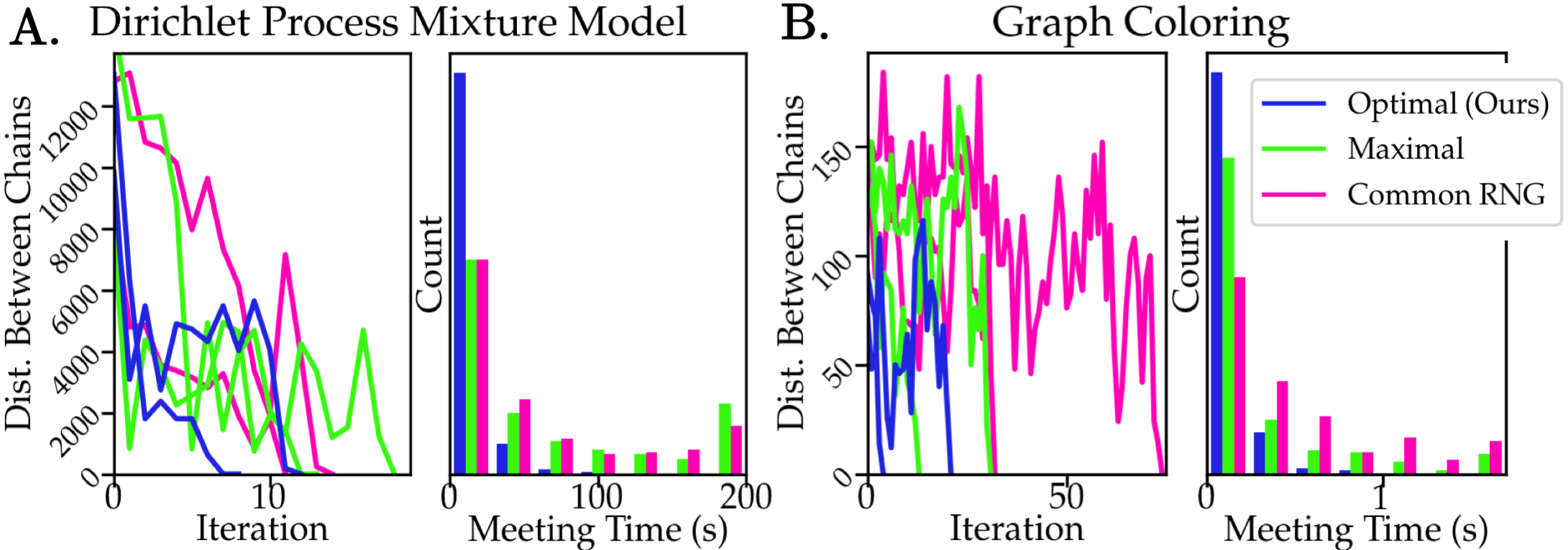}
    }
\end{figure}

\Cref{fig:toy_results} demonstrates that our approach yields faster couplings than the classical maximal coupling approach \citep[Section 5]{jerrum1998mathematical}, or an analogous coupling using shared common random numbers (see e.g.\ \citet{gibbs2004convergence}).
In applications to both Bayesian clustering and graph coloring, the distance between coupled chains stochastically decreases to $0$ (\Cref{fig:toy_results} left panels), with our approach leading to meetings after fewer sweeps.
Despite the larger per-sweep computational cost, our OT coupled chains typically meet after a shorter wall-clock time as well.
We suspect this improvement comes from avoiding label-switching, which hinders mixing of the maximal and common-RNG coupled chains.

The tightest bounds for mixing time for Gibbs samplers on graph colorings to date \citep{chen2019improved} rely on couplings on labeled representations.
Our results suggest better bounds may be attainable by considering convergence of partitions rather than labelings. 
Reducing the mixing time for Gibbs samplers of DPMM has been a motivation behind collapsed samplers \citep{maceachern1994estimating}, but the literature lacks upper bounds on the mixing time. 

\subsection{Unbiased estimation with parallel computation} \label{sec:unbiased-estimation}
We adapt the setup from \citet[Section 3.3]{jacob2020unbiased}.
Fixing a time budget, we run a single chain until time runs out and report the ergodic average. 
For coupled chains, we attempt as many meetings as possible in this time, and report the average across attempts. 

\paragraph{Posterior mean predictive density.}
The posterior predictive is a key quantity used in model selection \citep{gorur2010dirichlet}, 
and is of particular interest for DPMMs as it is known to be consistent for the underlying data distribution in total variation distance \citep{ghosal1999posterior}.
As a proof of concept, we computed unbiased estimates of the posterior predictive distribution of a DPMM (\Cref{fig:unbiased_estimation} A).

We generated $N = 100$ data points from a $10$-component Gaussian mixture model in one dimension, with the variance around cluster means equal to $4$.
We used a DPMM with $\alpha = 1$, $\mu_0 = 0$, $\Sigma_1 = 4.0$, $\Sigma_0 = 9.0$ to analyze the $N$ observations.
The solid blue curve is an unbiased estimate of the posterior predictive density. The black dashed curve is the true density of the population. The grey histogram bins the observed data. 
Because of the finite sample size, the predictive density is not equal to the true density.
In \Cref{apd:more-predictive}, the difference between the model's predictive density and the true density decreases as sample size $N$ increases.

\paragraph{Posterior mean component proportions.} 
A second key quantity of interest in DPMMs is the posterior mean of the proportion of data-points in the largest cluster(s) (e.g.\ as reported by \citet{liverani2015premium}).
We lastly explored parallel computation for unbiased estimation of this quantity on a real dataset (\Cref{fig:unbiased_estimation} B).
Specifically, we use a subset of the data used by \citet{prabhakaran2016dirichlet}, who used a DPMM to analyse single-cell RNA-sequencing data obtained from \citet{zeisel2015cell}
(see \Cref{apd:experimental_setup} for details).

\begin{figure}[t!]
	\floatconts
	{fig:unbiased_estimation}
	{\caption{
	Unbiased estimates for Dirichlet process mixture model are obtained using OT coupled chains. 
	(A) Unbiased estimate of the posterior predictive density for a toy problem.
        (B) Parallelism/accuracy trade-off for single and coupled chain estimators of the posterior mean portion of cells in the largest cluster.
        Each process is allocated 250 seconds, error bars indicate $\pm 2 \text{SEM}$.
        Ground truth denotes estimates from very long MCMC chains.
	}}
	{
		\includegraphics[width=0.95\linewidth]{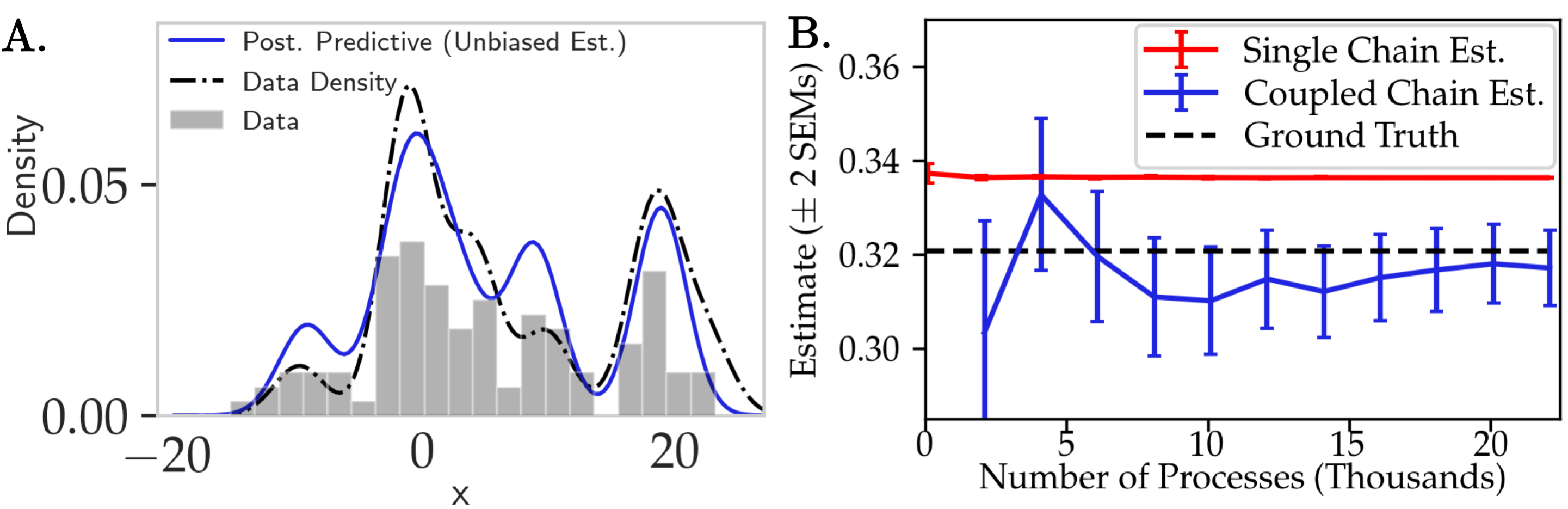}
	}
\end{figure}

\Cref{fig:unbiased_estimation} B presents a series of estimates of the proportion of cells in the largest component, and approximate frequentist confidence intervals. For each number of processes $M$, we aggregated $M$ independent single and coupled chain estimates,
each from a single processor with a {250} second limit.
We compare to the `ground-truth' proportion obtained by MCMC run for {10,000} sweeps.
Our results demonstrate the advantage of unbiased estimates in the high-parallelism, time-limited regime;
while single-chain estimates have lower variance, coupled chains yield smaller error when aggregated across many processes.
In addition, as result of unbiasedness, standard frequentist intervals may be expected to have good coverage.
By contrast, we cannot expect such intervals from single chains to be calibrated;
indeed, the true value is many standard errors from the single chain estimates (\Cref{fig:unbiased_estimation} B).

However, due to the variance of the unbiased estimates we require a degree of parallelism that may be impractical for most practitioners ($\approx$ {5,000} pairs of chains to attain error comparable to that of as many single chains). 
Indeed, in our experiments, we simulated this high parallelism by sequentially running batches of {100} processes in parallel.
Additionally, the estimation strategy can be finicky:
unbiasedness requires coupled chains to meet exactly, and for some models and experiments not shown,
we found that some pairs of coupled chains failed to meet quickly. 
This difficulty is expected for problems where single chains mix slowly, as slow mixing precludes the existence of fast couplings \citep[Chapter 3]{jacob2020couplings}.
Looking forward, we expect that our work will naturally benefit from advances in parallel-computation software and hardware, such as GPU implementations.  
Reducing the variance of the unbiased estimates is an open question, and is the target of ongoing work. 

\section*{Acknowledgements}
We thank the anonymous reviewers of AABI 2020 for pointing out a reference on other metrics between partitions \citep{meilua2007comparing}. 
This work was supported in part by an NSF CAREER Award and an ONR Early Career Grant.
BLT is supported by NSF GRFP.

\bibliography{references}

\appendix

\section{Proof of Gibbs Sweep Time Complexity}\label{apd:gibbs_sweep_proof}
We here detail our $O(N\maxK^3\log \maxK)$ implementation of \Cref{alg:gibbs}.
This serves as proof of \Cref{thm:Gibbs_complexity}.

Note that work in \Cref{alg:gibbs} may be separated into 2 computationally demanding stages for each of the $N$ data-points, $n \in [N]$;
computing the distances between each pair of partitions in the Cartesian product of supports of the Gibbs conditionals $\GibbsCond{N}{n}(\cdot | \pi_{-n})$ and $\GibbsCond{N}{n}(\cdot | \nu_{-n})$ and
solving the optimal transport problem in line 7.
As discussed in Remark \ref{rem:complexity}, the optimal transport problem may be solved in $O(K^3 \log K)$ time, and is the bottleneck step.
As such it remains only to show that for each $n\in [N]$, the pairwise distances may also be computed in $O(K^3\log K)$ time.

Recall that for two partitions $\pi, \nu \in \calP{N}$ the metric of interest is
\[
    \text{d}(\pi, \nu) 
    = \sum_{A\in \pi}|A|^2 
      +  \sum_{B \in \nu} |B|^2
      -2\sum_{(A, B) \in \pi\times\nu} |A\cap B|^2.
\]

However, it is not obvious from this expression alone that fast computation of pairwise distances should be possible.
We make this explicit in the following remark.
\begin{remark}
    Given constant $O(1)$ time for querying set membership (e.g.\ as provided by a standard hash-table set implementation),
for $\pi, \nu \in \mathcal{P}_N$, $d(\pi, \nu)$ in \Cref{eqn:distance_between_partitions} may be computed in $O\left(N \min(|\pi|, |\nu|) \right)$ time.
If we let $\maxK$ be the number of groups, so that $\maxK\approx \min(|\pi|, |\nu|)$, this gives
$O(N\maxK)$ time.
\end{remark}

While this is certainly faster than a naive approach relying on the formulation of this metric based on adjacency matrices, it is still not sufficient, as it is a factor of $N\maxK^2$ slower than the original (recall that we will need to do this for $\maxK^2$ pairs of clusters assignments).  

However we can do better for the Gibbs update by making two observations.
First, if we use $A_n$ and $B_n$ to denote the elements of $\pi$ and $\nu,$ respectively, containing data-point $n$, then for any $n$ we may write
\[\label{eqn:additional_distance}
    d(\pi,\nu) &= d(\pi_{-n}, \nu_{-n}) + 
    \left[|A_n|^2 - (|A_n|-n)^2 \right] + \left[|B_n|^2 - (|B_n|-n)^2 \right]  \\
    -2&\left[ |A_n \cap B_n|^2 - ( |A_n \cap B_n|-1)^2\right] \\
    &=d(\pi_{-n}, \nu_{-n}) + 2 \left[
        |A_n| + |B_n| - 2|A_n\cap B_n| \right].
\]

Second, the solution to the optimisation problem in \Cref{eq:minimal-Ed-coupling} is unchanged when we add a constant value to every distance:
Using again the notation of \Cref{alg:gibbs} we let $q:=\GibbsCond{N}{n}(\cdot | \pi_{-n})$ and $r:=\GibbsCond{N}{n}(\cdot | \nu_{-n})$ with supports 
$(\pi_1, \pi_2, \dots, \pi_K) = \text{support}(q)$ and 
$(\nu_1, \nu_2, \dots, \nu_{K^\prime}) = \text{support}(r)$.
and rewrite
\[\label{eqn:reframed_OT_problem}
    \gamma^* \defined \arg \min_{\gamma \in \Gamma(q, r)} \sum_{ \in \calP{N}} \sum_{y \in \calP{N}} d(x,y)\gamma(x,y) \\
=\arg \min_{\gamma \in \Gamma(q, r)} \sum_{x \in \calP{N}} \sum_{y \in \calP{N}} (d(x,y)-c)\gamma(x,y) 
\]
for any constant $c$; taking $c=d(\pi_{-n}, \nu_{-n})$ reveals that we need only compute the second term in \Cref{eqn:additional_distance}.

At first it may seem that this still does not solve the problem, as directly computing the size of the set intersections is $O(N)$ (if cluster sizes scale as $O(N)$).
However, \Cref{eqn:reframed_OT_problem} is just our final stepping stone.
If we additionally keep track of sizes of intersections at every step, updating them as we adapt the partitions will take constant time for each update.
As such, we are able to form the matrix of pairwise distances in $O(\maxK^2)$ time.
Regardless of $N$, this moves the bottleneck step to solving the OT problem which, as discussed in \Cref{rem:complexity}, may be computed in $O(\maxK^3 \log \maxK)$ time with Orlin's algorithm \citep{orlin1993faster}.
We provide a practical implementation of this approach in our code; see \texttt{pairwise\_dists()} in \texttt{modules/utils.py}.

\section{Additional Experimental Details}\label{apd:experimental_setup}
\subsection{Meeting time distributions}

\paragraph{DP mixtures.}
For each replicate, we simulated $N=150$ data-points from a $K=4$ component, 2 dimensional Gaussian mixture model.
The target distribution was the posterior of the probabilistic model \Cref{eq:dpmm-model}, with $\Sigma_0=2.5 I_2$, $\Sigma_1=2 I_2$ and $\alpha=0.2$.
For each replicate true means for the finite mixture were sampled as $\mu_k \sim \mathcal{N}(0, \Sigma_0)$, mixing proportions as $\theta \sim \text{Dir}(\alpha 1_K)$,
and each of the $n\in[N]$ observations as $z_n\sim \text{Cat}(\theta)$, $\data_n\sim \mathcal{N}(\mu_{z_n}, \Sigma_1)$.
See \texttt{notebooks/Coupled\_CRP\_sampler.ipynb} for complete implementation and details. 
This code is adapted from \href{https://github.com/tbroderick/mlss2015_bnp_tutorial/blob/master/ex5_dpmm.R}{github.com/tbroderick/mlss2015\_bnp\_tutorial/blob/master/ex5\_dpmm.R}

\paragraph{Graph coloring}
Let $G$ be an undirected graph with vertices $V=[N]$ and edges $E \subset V \otimes V,$ and let $Q=[q]$ be set of $q$ colors.
A graph coloring is an assignment of a color in $Q$ to each vertex satisfying that the endpoints of each edge have different colors.
We here demonstrate an application of our method to a Gibbs sampler which explores the uniform distribution over valid $q-$colorings of $G$,
i.e.\ the distribution which places equal mass on ever proper coloring of $G$.

To employ \Cref{alg:gibbs}, for this problem we need only to characterise the PMF on partitions of the vertices implied by the uniform distribution on its colorings.

A partition corresponds to a proper coloring only if no two adjacent vertices are in the element of the partition.
As such, we can write
$$
p_{\Pi_N}(\pi) \propto \ind{|\pi| \le q \text{  and  }  A(\pi)_{i,j}=1 \rightarrow (i, j) \not \in E, \ \forall i \ne j } { q \choose |\pi|} |\pi|!,
$$
where the indicator term checks that $\pi$ can correspond to a proper coloring
and the second term accounts for the number of unique colorings which induce the partition $\pi$.
In particular it is the product of the number of ways to choose $|\pi|$ unique colors from $Q$ ( ${q \choose |\pi|} :=\frac{q!}{|\pi|! (q-|\pi|)!}$) and the number of ways to assign those colors to the groups of vertices in $\pi$.

For the experiments in \Cref{fig:toy_results}, we simulated Erd\H{o}s-R\'{e}nyi random graphs with $N=25$ vertices, and including each possible edge with probability $0.2$.
We chose a maximum number of colors $Q$ by first initializing a coloring greedily and setting $Q$ as the number of colors used in this initial coloring plus two.
See \texttt{notebooks/coloring\_OT.ipynb} for complete implementation and results.
This code is adapted from: \\
\href{https://github.com/pierrejacob/couplingsmontecarlo/blob/master/inst/chapter3/3graphcolourings.R}{github.com/pierrejacob/couplingsmontecarlo/inst/chapter3/3graphcolourings.R}

\subsection{Unbiased estimation}

\paragraph{Predictive density in Gaussian mixture data.}
The true density is a $10$-component Gaussian mixture model with known observational noise variance $\sigma = 2.0$. 
The cluster proportions were generated from a symmetric Dirichlet distribution with mass $1$ for all $10$-coordinates. 
The cluster means were randomly generated from $\distNorm{0}{10^2}$. 
The target DP mixture model had $\alpha = 1$, standard deviation over cluster means $3.0$ and standard deviation over observations $2.0$.
The function of interest is the posterior predictive density
\begin{equation} \label{eq:pred-def}
	\Pr(\data_{N+1} \in dx \given \data_{1:N}) = \sum_{\Pi_{N+1}} \Pr(\data_{N+1} \in dx \given \Pi_{N+1}, \data_{1:N}) \Pr(\Pi_{N+1} \given \data_{1:N}).
\end{equation}
In \cref{eq:pred-def}, $\Pi_{N+1}$ denotes the partition of the data $\data_{1:(N+1)}$. To translate \cref{eq:pred-def} into an integral over just the posterior over $\Pi_{N}$, the partition of $\data_{1:N}$, we break up $\Pi_{N+1}$ into $(\Pi_{N}, Z)$ where $Z$ is the cluster indicator specifying the cluster of $\Pi_{N}$ (or a new cluster) to which $\data_{N+1}$ belongs. Then
\begin{equation*}
	\Pr(\data_{N+1} \in dx \given \data_{1:N}) = \sum_{\Pi_N} \left[ \sum_{Z}  \Pr(\data_{N+1} \in dx, Z \given \Pi_{N}, \data_{1:N}) \right] \Pr(\Pi_N \given \data_{1:N})
\end{equation*}
Each $\Pr(\data_{N+1} \in dx, Z \given \Pi_{N}, \data_{1:N})$ is computed using the prediction rule for the CRP and Gaussian conditioning. Namely
\begin{equation*}
	\Pr(\data_{N+1} \in dx, Z \given \Pi_{N}, \data_{1:N}) = \underbrace{\Pr(\data_{N+1} \in dx \given Z, \Pi_{N}, \data_{1:N})}_{\text{Posterior predictive of Gaussian}} \times \underbrace{\Pr(Z \given \Pi_{N})}_{\text{CRP prediction rule}}.
\end{equation*}
The first term is computed with the function used during Gibbs sampling to reassign data points to clusters. In the second term, we ignore the conditioning on $\data_{1:N}$, since $Z$ and $\data_{1:N}$ are conditionally independent given $\Pi_{N}.$

We ran {10,000} replicates of the time-budgeted estimator using coupled chains, each replicate given a sufficient time budget so that all {10,000} replicates had at least one successful meeting in the allotted time. 

\paragraph{Top component proportion in single-cell RNAseq.}
We extracted $D=50$ genes with the most variation of $N=200$ cells.
We then take the log of the features, and normalize so that each feature has mean $0$ and variance $1$.
We as our target the posterior of the probabilistic model in \cref{eq:dpmm-model} with $\alpha = 1.0$, $\mu_0 = 0$, $\Sigma_0 = 0.5$, $\Sigma_1 = 1.3 I_D$.
Notably, this is a simplification of the set-up considered by \citet{prabhakaran2016dirichlet}, who work with a larger dataset and additionally perform fully Bayesian inference over these hyper-parameters.
In our experiments, the function of interest is the posterior expected of the proportion of cells in the largest cluster i.e.\ $\mathbb{E}[\max_{|A| \in \pi}|A|/N \given \data ].$

\section{More plots of predictive density} \label{apd:more-predictive}
\subsection{Posterior concentration implies convergence in total variation of predictive density}

Some references on posterior concentration are \citet{ghosal1999posterior,lijoi2005consistency}. 
The true data generating process is that there exists some density $f_0$ w.r.t.\ Lebesgue measure that generates the data in an iid manner $X_1,X_2,\ldots,X_n$. 
We use the notation $P_{f_0}$ to denote the probability measure with density $f_0$. 
The probabilistic model is that we have a prior $\Pi$ over densities $f$, and observations $X_i$ are conditionally iid given $f$. 
Let $\mathcal{F}$ be the set of all densities on $\mathbb{R}$.
For any measurable subset $A$ of $\mathcal{F}$, the posterior of $A$ given the observations $X_i$ is denoted $\Pi(A|X_{1:N})$. 
A strong neighborhood around $f_0$ is any subset of $\mathcal{F}$ containing a set of the form $V = \{f \in \mathcal{F}: \int | f - f_0| < \epsilon \}$ according to \citet{ghosal1999posterior}.
The prior $\Pi$ is strongly consistent at $f_0$ if for any strong neighborhood $U$,
\begin{equation} \label{eq:consistency}
	 \lim_{n \to \infty} \Pi(U|X_{1:n})  = 1,
\end{equation}
holds almost surely for $X_{1:\infty}$ distributed according to $P_{f_0}^{\infty}$. 

\begin{theorem} [{\citet[Proposition 4.2.1]{ghosh2003bayesian}}] \label{thm:consistent-predictive}
	If a prior $\Pi$ is strongly consistent at $f_0$ then the predictive distribution, defined as
	\begin{equation} \label{eq:pred-as-density}
		\widehat{P}_n(A \mid X_{1:n}) \coloneqq \int_f P_f(A)  \Pi(f \mid X_{1:n})
	\end{equation}
	 also converges to $f_0$ in total variation in a.s. $P_{f_0}^{\infty}$
	\begin{equation*}
		d_{TV}\left( \widehat{P}_n, P_{f_0}\right) \xrightarrow{} 0.
	\end{equation*}
\end{theorem}

The definition of posterior predictive density in \cref{eq:pred-as-density} can equivalently be rewritten as
\begin{equation*}
	\widehat{P}_n(A \mid X_{1:n}) = \Pr(X_{n+1} \in A \given X_{1:n}),
\end{equation*}
since $P_f(A) = P_f(X_{n+1} \in A)$ and all the $X$'s are conditionally iid given $f$.

\begin{theorem}[{DP mixtures prior is consistent for finite mixture models}] \label{thm:DPMM-consistent-finite-mixture}
	Let the true density be a finite mixture model $f_0(x) \defined \sum_{i=1}^{m} p_i \mathcal{N}(x|\theta_i,\sigma_1^2)$. Consider the following probabilistic model
	\begin{align*}
		\widehat{P} &\sim \distDP{\alpha}{\mathcal{N}(0,\sigma_0^2)} \\
		\theta_i \mid \widehat{P} &\stackrel{iid}{\sim} \widehat{P} & & i = 1,2,\ldots,n \\
		X_i \mid \theta_i &\stackrel{indep}{\sim} \mathcal{N}(\theta_i,\sigma_1^2) & & i = 1,2,\ldots,n
	\end{align*}
	Let $\widehat{P}_n$ be the posterior predictive distribution of this generative process. Then with a.s.\ $P_{f_0}$
	\begin{equation*}
		d_{TV}\left( \widehat{P}_n, P_{f_0}\right) \xrightarrow{n \to \infty} 0.
	\end{equation*}
\end{theorem}

\begin{proof}[{Proof of \Cref{thm:DPMM-consistent-finite-mixture}}]
First, we can rewrite the DP mixture model as a generative model over continuous densities $f$ 
\begin{equation} \label{eq:dpmm-as-density}
\begin{aligned}
	\widehat{P} &\sim \distDP{\alpha}{\mathcal{N}(0,\sigma_0^2)} \\
	f &= \mathcal{N}(0,\sigma_1^2) \ast \widehat{P} 
	\\
	X_i \mid f &\stackrel{iid}{\sim} f & & i = 1,2,\ldots,n
\end{aligned}
\end{equation}
where $\mathcal{N}(0,\sigma_1^2) \ast \widehat{P}$ is a convolution, with density $f(x) \defined \int_{\theta} \mathcal{N}(x - \theta | 0,\sigma_1^2) d\widehat{P}(\theta)$. 
	
The main idea is showing that the posterior $\Pi(f|X_{1:n})$ is strongly consistent and then leveraging \Cref{thm:consistent-predictive}. For the former, we verify the conditions of \citet[Theorem 1]{lijoi2005consistency}. 

The first condition of \citet[Theorem 1]{lijoi2005consistency} is that $f_0$ is in the K-L support of the prior over $f$ in \Cref{eq:dpmm-as-density}. We use \citet[Theorem 3]{ghosal1999posterior}. Clearly $f_0$ is the convolution of the normal density $\mathcal{N}(0,\sigma_1^2)$ with the distribution $P(.) = \sum_{i=1}^m p_i \delta_{\theta_i}$. $P(.)$ is compactly supported since $m$ is finite. Since the support of $P(.)$ is the set $\{\theta_i\}_{i=1}^{m}$ which belongs in $\mathbb{R}$, the support of $\mathcal{N}(0,\sigma_0^2)$, by \citet[Theorem 3.2.4]{ghosh2003bayesian}, the conditions on $P$ are satisfied. The condition that the prior over bandwidths cover the true bandwidth is trivially satisfied since we perfectly specified $\sigma_1$. 

The second condition of \citet[Theorem 1]{lijoi2005consistency} is simple: because the prior over $\widehat{P}$ is a DP, it reduces to checking that
\begin{equation*}
	\int_{\mathbb{R}} |\theta| \mathcal{N}(\theta \mid 0, \sigma_0^2) < \infty
\end{equation*}
which is true. 

The final condition trivial holds because we have perfectly specified $\sigma_1$: there is actually zero probability that $\sigma_1$ becomes too small, and we never need to worry about setting $\gamma$ or the sequence $\sigma_k$.  
\end{proof}

\subsection{Predictive density plots for varying N}

In \Cref{fig:predictives}, the distance between the posterior predictive density and the underlying density decreases as $N$ increases. 
We sampled a grid $\{u_j\}$ of $150$ evenly-spaced points in the domain $[-20, 30]$, and evaluated both the true density and the posterior predictive density on this grid.
The distance in question sums over the absolute differences between the evaluations over the grid
\begin{equation*}
	\text{dist} \defined \sum_{j} | f_N(u_j) - f_0(u_j)|. 
\end{equation*}
where $f_N(u_j)$ is the posterior predictive density of the $N$ observations under the DPMM at $u_j.$
The distance is meant to illustrate \emph{pointwise} rather than total variation convergence.
Although the predictive density converges in total variation to the underlying density, it is only guaranteed that a subsequence of the predictive density converges pointwise to the underlying density. 

\begin{figure}[htbp]
	\floatconts
	{fig:predictives}
	{\caption{Posterior predictive density for different $N$. The time budget for each replicate when $N = 100, 200,300$ is respectively $100, 300, 800$ seconds. We average the results from $400$ replicates.}}
	{%
		\subfigure[$N = 100$.]{\label{sub-fig:predictivendata100}%
			\includegraphics[width=0.33\linewidth]{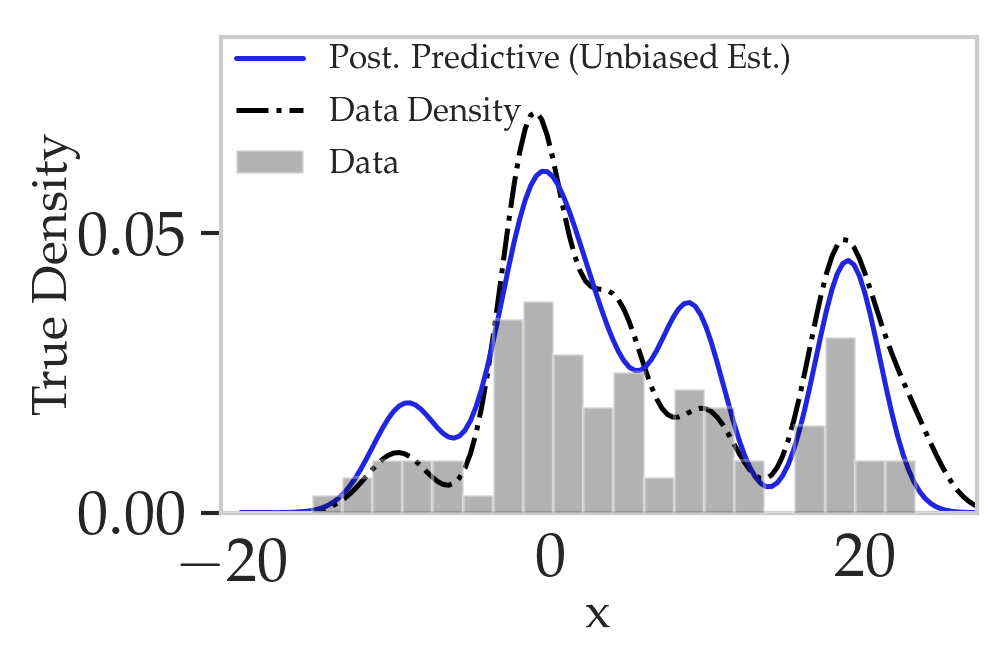}}%
		\subfigure[$N = 200$.]{\label{sub-fig:predictivendata200}%
			\includegraphics[width=0.33\linewidth]{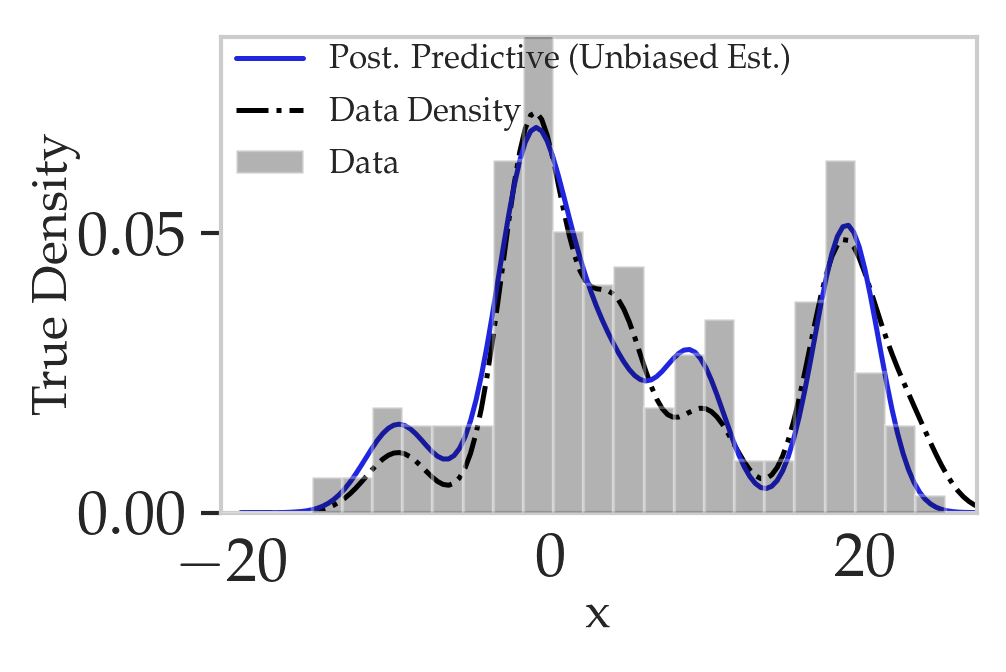}}%
		\subfigure[$N = 300$.]{\label{sub-fig:predictivendata300}%
			\includegraphics[width=0.33\linewidth]{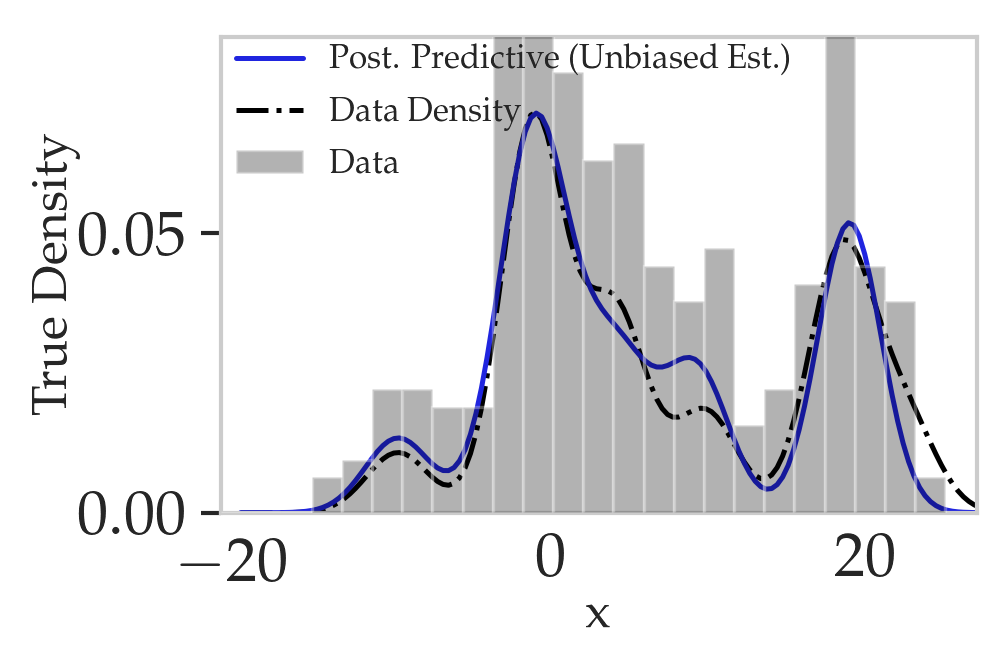}}%
	}
\end{figure}
In \Cref{fig:predictives}, each $N$ has a different time budget because for larger $N$, in general per-sweep time increases and number of sweeps until coupled chains meet also increase.

\end{document}